\documentclass[sigconf]{acmart}
\AtBeginDocument{%
  \providecommand\BibTeX{{%
    Bib\TeX}}}

\setcopyright{none}
\renewcommand\footnotetextcopyrightpermission[1]{}

\newtheorem{theorem}{Theorem}
\usepackage[abbreviations]{glossaries-extra}
\usepackage{soul}
\usepackage{subcaption}
\usepackage{booktabs}
\usepackage{multirow}
\usepackage{rotating}
\makeglossaries
\usepackage{xcolor}
\usepackage{amsfonts}
\usepackage{mathtools}
\usepackage{algorithmic}
\usepackage[colorinlistoftodos]{todonotes}
\def\BibTeX{{\rm B\kern-.05em{\sc i\kern-.025em b}\kern-.08em
    T\kern-.1667em\lower.7ex\hbox{E}\kern-.125emX}}
\usepackage{tikz}
\usetikzlibrary{arrows.meta}

\begin{document}
\renewcommand{\dbltopfraction}{0.9}
\renewcommand{\dblfloatpagefraction}{0.8}
\setcounter{dbltopnumber}{2}

\title{MQSS-Selector: RL-Guided Pass Selection for an MLIR Compilation Pipeline}
\titlenote{Code available at:
  \url{https://anonymous.4open.science/r/Pass_selector-10AE/README.md}}

\author{A. Youssefi}
\affiliation{%
  \institution{Leibniz Supercomputing Centre (LRZ)}
  \city{Garching bei München}
  \country{Germany}}
\email{youandre100@googlemail.com}

\author{E. Kaya}
\affiliation{%
  \institution{Leibniz Supercomputing Centre (LRZ), Technical University of Munich (TUM)}
  \city{Garching bei München}
  \country{Germany}}
\email{ercument.kaya@tum.de}

\author{M. Chung}
\affiliation{%
  \institution{Leibniz Supercomputing Centre (LRZ)}
  \city{Garching bei München}
  \country{Germany}}
\email{minh.chung@lrz.de}

\author{J. Echavarria}
\orcid{0000-0002-3751-5273}
\affiliation{%
  \institution{Munich Quantum Valley (MQV)}
  \city{Garching bei München}
  \country{Germany}}
\email{jorge.echavarria@lrz.de}

\author{L. B. Schulz}
\orcid{0000-0002-4702-3440}
\affiliation{%
  \institution{Argonne National Laboratory (ANL)}
  \city{Lemont, IL}
  \country{United States}}
\email{schulz@anl.gov}

\author{M. Schulz}
\orcid{0000-0001-9013-435X}
\affiliation{%
  \institution{Leibniz Supercomputing Centre (LRZ), Technical University of Munich (TUM)}
  \city{Garching bei München}
  \country{Germany}}
\email{schulzm@in.tum.de}

\newabbreviation{api}{API}{Application Programming Interface}
\newabbreviation{a3c}{A3C}{Asynchronous Advantage Actor-Critic}
\newabbreviation{cpu}{CPU}{Central Processing Unit}
\newabbreviation{dl}{DL}{Deep Learning}
\newabbreviation{dse}{DSE}{Design Space Exploration}
\newabbreviation{eqs3}{EQS3}{European Quantum Systems and Software Summit}
\newabbreviation{gpu}{GPU}{Graphics Processing Unit}
\newabbreviation{gae}{GAE}{General Advantage Estimation}
\newabbreviation{hpc}{HPC}{High Performance Computing}
\newabbreviation{hpcqc}{HPCQC}{High Performance Computing-Quantum Computing}
\newabbreviation{ir}{IR}{Intermediate Representation}
\newabbreviation{isv}{ISV}{Independent Software Vendor}
\newabbreviation{jsd}{JSD}{Joint-Schedule Decision Problem}
\newabbreviation{jit}{JIT}{Just-In-Time}
\newabbreviation{knn}{KNN}{K Nearest Neighbors}
\newabbreviation{lrz}{LRZ}{Leibniz Supercomputing Centre}
\newabbreviation{lstm}{LSTM}{Long Short Term Memory}
\newabbreviation{mae}{MAE}{Mean Absolute Error}
\newabbreviation{ml}{ML}{Machine Learning}
\newabbreviation{mlir}{MLIR}{Multi-Level Intermediate Representation}
\newabbreviation{mlp}{MLP}{Multi-Layer Perceptrons}
\newabbreviation{mqss}{MQSS}{Munich Quantum Software Stack}
\newabbreviation{mqt}{MQT}{Munich Quantum Toolkit}
\newabbreviation{mqv}{MQV}{Munich Quantum Valley}
\newabbreviation{mse}{MSE}{Mean Square Error}
\newabbreviation{nisq}{NISQ}{Noisy Intermediate-Scale Quantum}
\newabbreviation{qc}{QC}{Quantum Computing}
\newabbreviation{qdessi}{Q-DESSI}{Quantum Development Environment, System Software \& Integration}
\newabbreviation{qec}{QEC}{Quantum Error Correction}
\newabbreviation{qdmi}{QDMI}{Quantum Device Management Interface}
\newabbreviation{qir}{QIR}{Quantum Intermediate Representation}
\newabbreviation{qpi}{QPI}{Quantum Programming Interface}
\newabbreviation{qpu}{QPU}{Quantum Processing Unit}
\newabbreviation{qrm}{QRM}{Quantum Resource Manager}
\newabbreviation{qrmi}{QRMI}{Quantum Resource Management Interface}
\newabbreviation{qnn}{QNN}{Quantum Neural Network}
\newabbreviation{rf}{RF}{Random Forest}
\newabbreviation{rl}{RL}{Reinforcement Learning}
\newabbreviation{rnn}{RNN}{Recurrent Neural Networks}
\newabbreviation{svm}{SVM}{Support Vector Machine}
\newabbreviation{tsp}{TSP}{Traveling Salesman Problem}
\newabbreviation{td}{TD}{Temporal Difference}
\newabbreviation{vqe}{VQE}{Variational Quantum Eigensolver}

\begin{abstract}
\gls{hpc} and \gls{qc} systems are increasingly converging towards unified \gls{hpcqc} infrastructures, driven by a growing need to bridge classical and quantum workflows, which affects all levels of the system stack, from the hardware to compilers and runtimes, all the way to applications.
However, today's \gls{qc} devices are still in the \gls{nisq} era, are error-prone and resource-limited, and therefore require specialized optimizations and topology mappings to achieve sufficient fidelity. This places special emphasis on proper compilation and optimization within the overall quantum software stack.
Many existing stacks remain fragmented, with separate components responsible for device selection, compiler-pass optimization, and job queue scheduling.
This paper proposes a unified, learning-based selector that integrates these disparate stages into a cohesive framework.
Our proposed selector scheme leverages reinforcement learning and deep learning models that can be extended to simultaneously optimize multiple objectives---such as fidelity, compilation time, and scheduling latency---while dynamically adapting to circuit characteristics and device conditions.
\end{abstract}



\maketitle
\pagestyle{plain}
\thispagestyle{plain}

\glsresetall

\section{Introduction} 
\label{sec:intro}

As \gls{hpc} and \gls{qc} systems converge, the design of software infrastructures that support seamless integration of classical and quantum workloads becomes critical.
In hybrid \gls{hpcqc} environments, quantum processors are typically accessed as off-node accelerators, where quantum programs are offloaded for execution.

The execution pipeline requires appropriate scheduling and quantum device selection that can select the right resources, coordinate classical-quantum tasks, and optimize performance. State-of-the-art heterogeneous \gls{hpcqc} stacks, such as the \gls{hpcqc} efforts at RIKEN in connection with IBM~\cite{swayne2025ibm} devices, the software stacks originating from the Chicago Quantum Exchange~\cite{cqe2025advancing}, the Quantum Delta in the Netherlands~\cite{10.1145/2903150.2906827}, or the \gls{mqss} developed in the \gls{mqv} initiative~\cite{10313717}, adopt one- or even two-level scheduling approaches. Specifically, the first step is with a \gls{hpc} scheduler (e.g., SLURM~\cite{yoo2003slurm}) that allocates classical resources, followed by a second-level quantum scheduler for selecting the target device and optimizing the task for the device.

The second step is typically compilation and optimization for quantum
circuits. We employ MLIR~\cite{lattner2020mlir}, a compiler framework that
handles multiple IRs on different levels via the abstraction of dialects,
which makes it extensible to new domains. Program modifications and analyses
are implemented as passes, which are composable and reusable components. In
the context of our work, we define a pass as an operation applied to a quantum
program that either transforms it into an equivalent representation or
performs a data analysis on the job. In contrast to classical systems,
applying quantum compilation passes can be done at runtime, as critical input
parameters are not known statically.

\begin{figure*}[t!]
  \centering
  \includegraphics[width=\textwidth]{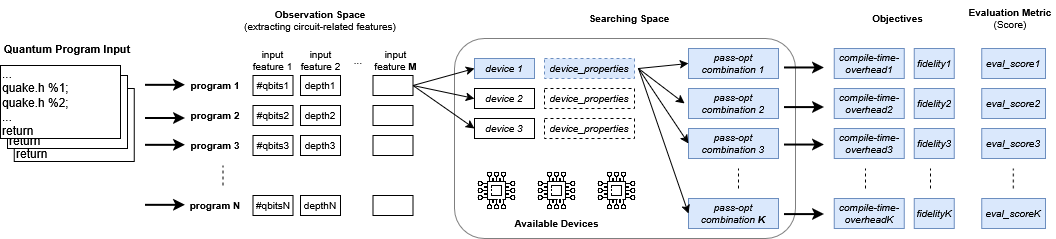}
  \caption{Illustration of the scheduling pipeline integrating \textit{device selection} and \textit{phase ordering} (i.e., compiler pass-sequence optimization) for a quantum program execution.}
  \label{fig:problem_space}
\end{figure*}

The current challenges of \gls{nisq} devices are their inherently error-prone nature, devices are constrained in qubit count, and highly sensitive to calibration and workload conditions~\cite {preskill2018nisq}.
The software environments must be carefully designed in order to mitigate these limitations to the extent possible.
In particular, quantum programs must be compiled and optimized prior to execution to tailor them to the selected system, tune to the system's topology, and account for the current calibration data.
Compilation and optimization, therefore, involve a complex software stack capable of achieving  Pareto-optimal quantum device selection, while also ensuring efficient and robust \gls{jit} compilation and optimization. 
Generally, as mentioned, this is achieved through a series of unitary-preserving transformations or passes.

Device and compilation pass selection is considered a critical challenge, as both involve two NP-hard problems, namely device selection and phase ordering problems~\cite {dangwal2024compass,wang2024phaseorderingproblemfinding}. These two problems can be seen as isolated challenges. In which, phase ordering implies a valid and effective order of passes for transforming quantum programs.
Conventionally, related works often apply heuristics or static policies without coordination across stages~\cite{MQT_Pred}.

Crucially, device selection and phase ordering problems are both inherently \emph{multi-objective combinatorial problems} with tightly \emph{interdependent} design spaces:
On the one hand, device selection involves choosing a target device from a set of available quantum devices, where each device may differ in topology, gate fidelity, queuing delay, and scheduled maintenance windows.
The selection's objectives may include, for example, maximizing compatibility between an arbitrary quantum program and a chosen quantum device, maximizing device availability, minimizing queue wait time, and meeting time-varying availability constraints based on device workload and maintenance schedules.
On the other hand, phase ordering is similarly a multi-objective problem: choosing a sequence from a set of compilation and optimization passes, potentially with repetition, to minimize \gls{jit} compilation and execution runtime, and to produce a program of adequate quality and device compliance.
We propose a unified scheduling framework called MQSS-Selector, which addresses the stages of the quantum execution pipeline from device selection to program transformation together. Scheduling in this context focuses on selecting quantum devices and compilation passes. 
We use a \gls{rl}-based approach to develop an adaptive selector. 
The selector's objective is to make decisions based on program characteristics, available compilation passes and quantum backends.

The contributions of this paper are as follows:
\begin{itemize}
    \item \textbf{ MQSS-Selector}: an \gls{rl}-based \textit{Pass-Selector} module that learns to choose and order compilation passes for quantum circuits, demonstrating the feasibility of data-driven, MLIR-native pass selection --- in contrast to prior work, which largely targets Qiskit-based pass pipelines.
    \item \textbf{Dual-Annealed Exploration Priming (DAEP)}: a training strategy for the reinforcement learning agent that improves efficiency during training.
    \item A set of analysis metrics for the agent's actor and critic networks, enabling evaluation of the learned policy independently of other compiler tools.
    \item An extension of the \gls{mqss} X Compilation Pass Suite with a considerably larger set of \gls{mlir}-based passes.
   
\end{itemize}

\section{Problem Statement and Motivation}\label{sec:prob_statement}

The theoretical problem statement can be understood by the theorem below.

\begin{theorem}\label{theorem:partition}
The problems \textit{Device selection} and \textit{Phase ordering} are NP-hard, and they both reduce to an NP‑hard combinatorial optimization problem and jointly determine Pareto‑optimal compilation results.
\end{theorem}

\begin{proof}
    We use the classic \textit{Partition} problem---that is, deciding whether a multiset of positive integers can be split into two subsets of equal sum---as our hardness base. 
    \textit{Partition} is well known to be NP‐complete~\cite{mertens2003easiesthardproblemnumber}.

    \begin{enumerate}
        \item \textit{Device selection is NP‑hard}: The task of choosing which quantum accelerator to use for a set of programs can be described as a scheduling problem---each program (job) has a ``cost'' (e.g., fidelity, runtime, or resource use) and each device has limited capacity.
        Assigning an arbitrary program %
        to a target device %
        for optimizing makespan or throughput is a classical scheduling problem, known to be NP-hard~\cite{Seitz_2024}.
        In particular, we can reduce the NP-complete \textit{Partition} problem to device selection: given integers $a_1,\dots,a_n$. By normalizing them with the linear function $a_i^*=a_i/a_{max}$, the values get mapped to the interval $[0,1]$. Since $a_{max}$ is part of the input, encoding each normalized value $a_i^*$ ---and hence the calibration of the associated operation--- requires precision logarithmic in $a_{max}$, so the constructed programs remain polynomial in the input size. Analogously to the $n$ integers create  $n$ programs whose costs $1-F_H~ \eqref{eq:hellinger-fidelity}$ are $a_1^*,\dots,a_n^*$. We use two identical quantum devices as the partitioned multi-sets. By assuming theoretical devices with $n$ arbitrary calibrated native operations, one can construct the necessary cost values. Each program for a value is then just the associated calibrated operation. Furthermore, we assume that all programs have the same fidelity across the devices. While this does not hold physically, it describes a valid special case.
        Deciding whether the programs can be split into two subsets with total cost exactly $\sum_i a^*_i/2$ is the \textit{Partition} question.
        Thus, optimal device allocation solves \textit{Partition}, implying \textit{device selection} is NP-hard.
        
        \item \textit{Phase or pass ordering problem is NP‑hard}: Finding the best order of compiler optimization passes is itself a classic NP-hard problem.
        Wang et al. prove it by reducing Phase ordering to the halting problem and also prove undecidability~\cite{wang2024phaseorderingproblemfinding}.

    \end{enumerate}

    Both reductions show that device selection and phase ordering are NP-hard. The combined problem is also NP-hard, as device selection is the special case with only one available phase combination, and vice versa; limiting to only one device resolves to the Phase ordering problem.
    Therefore, the combined problem is at least as hard as NP-hard.

\end{proof}

In practice, this NP-hardness translates into a range of hardware-dependent challenges in quantum program compilation, including:

\begin{itemize}
    \item Translating input program to a required exchange format.
    \item Optimizing the overall program depth respecting coherence constraints from the device.
    \item Decomposing quantum instructions into native gate sets.
    \item Mapping the program to device-specific qubit connectivity constraints.
    \item Reducing the total number of operations in a program.

\end{itemize}

Managing these tasks requires applying various compilation and optimization passes that are compatible with a target quantum device.
However, not all passes are beneficial for arbitrary quantum programs; in fact, certain combinations of passes can degrade performance. 
Therefore, it is necessary to treat each quantum job as an instance of a specific program class and to determine a tailored pass sequence that optimizes compilation while yielding a quantum program compliant with the target device's underlying architecture.

The task of the selector in quantum program execution is effectively a \gls{dse}, as illustrated in Fig.~\ref{fig:problem_space}. Given an input quantum program, we assume that relevant circuit-level features are extracted prior to making decisions about \textit{device selection} or \textit{phase ordering}. These features define the design space that may include metrics such as the number of qubits, program depth, and other static properties, as proposed by Tomesh~et~al. in SupermarQ~\cite{tomesh2022supermarq}. Based on this design space, a unified scheduler must choose a target quantum accelerator (\textit{device selection}) and a suitable sequence of compilation passes (\textit{phase ordering}), together forming the search space. Each combination of decisions leads to outcomes characterized by execution objectives, such as compilation time and expected fidelity. These objectives can be treated separately or aggregated into a unified evaluation metric for decision-making. In general, we formalize the joint \textit{device selection} and \textit{phase ordering} across the following key dimensions:
\begin{enumerate}
    \item \emph{Jobs}: Incoming quantum programs with varying structures and resource requirements.
    \item \emph{Device Specs}: Dynamic backend information including queue lengths, calibration data, and noise characteristics.
    \item \emph{Passes}: Available compiler transformations and program optimization strategies.
    \item \emph{Objectives}: Measurable performance indicators such as circuit depth, execution fidelity.
\end{enumerate}

As shown in \textsc{Theorem}~\ref{theorem:partition}, the proposed joint \textit{device selection} and \textit{phase ordering} scheme can be mapped to a variant of the \textit{Partition} problem. 
Moreover, since the \textit{Partition} problem is NP‑complete, exact solutions do not scale, whereas heuristic and approximation algorithms can efficiently find Pareto-optimal candidate solutions.
This interdependence between the \textit{device selection} and the \textit{phase ordering} problems is also referred to as a bootstrapping dilemma, and is also known in programming languages~\cite{reynolds:bootstrapping2003}.
The mutual dependency complicates modular optimization and motivates our proposal for a unified selector that jointly reasons over both device compatibility and compiler pass strategies within a single decision-making framework.
In this paper, we argue that such a unified approach, driven by \gls{rl} models, is essential to navigating the multi-objective and dynamically changing quantum execution landscape. This selector is toward Pareto-efficient configurations that outperform sequential or isolated heuristics.

\section{Related Work} \label{sec:related_work}
A recent survey illustrates the wide range of applications for \gls{rl}
algorithms in the domain of quantum technology~\cite{bukov2026rl}.
\gls{rl} algorithms are already used to learn compiler optimization decisions
in production compilers for classical programs, as demonstrated by MLGO~\cite{mlgo} for LLVM.
Another such project is AutoPhase~\cite{autophase}, in which an \gls{rl} agent
learns to solve the phase ordering problem for LLVM.
Bacher~et~al.~\cite{qrrms} discuss the gap in the treatment of quantum
resources within \gls{hpc} resource management systems and introduce a Slurm
plugin that makes quantum backends schedulable alongside classical compute
nodes.
Kulkarni and Chaudhary~\cite{kulkarni_dist_sched} show the importance of
scheduling in distributed quantum systems.
It is for such settings that we develop the unified scheduler, which
distributes and compiles a task such that, for example, the expected fidelity
is maximized.
Furthermore, several studies have explored \gls{ml} and \gls{rl} approaches to optimize quantum circuit compilation and device selection.
Fösel~et~al.~\cite{fösel2021quantumcircuitoptimizationdeep} propose a deep \gls{rl} agent that learns hardware‐aware circuit transformations, achieving notable reductions in gate count and circuit depth. 
Quetschlich~et~al. address similar challenges in two complementary works: first, by predicting favorable compilation options via supervised learning~\cite{Quetschlich_pred}, and second, by employing \gls{rl} to discover optimized pass sequences across multiple compilers~\cite{Quetschlich_compile}.
TuniQ~\cite{tuniq} is an \gls{rl}-based project for compilation that accounts for compilation time and fidelity, but it is based on Qiskit.
TuniQ follows a divide-and-conquer strategy that separates the phase ordering problem into subproblems, addressing mapping, decomposition, optimization, and further compilation stages, thereby limiting the available passes at each stage.

More recently, the \textit{\gls{mqt} Predictor} framework proposed \textit{device selection} with \textit{phase ordering} customization, demonstrating improved fidelity and reduced depth performance on diverse quantum backends~\cite{MQT_Pred}.
\textit{\gls{mqt} Predictor} focuses on a hybrid scheduling pipeline where \textit{device selection} and \textit{phase ordering} optimization are treated as two separate learning tasks.
A supervised model is trained for \textit{device selection}, while a distinct \gls{rl} agent is trained for \textit{phase ordering}. 
While effective in their respective domains, these components operate independently and are trained in isolation. 
Furthermore, \textit{\gls{mqt} Predictor} does not incorporate dynamic backend conditions, such as queue length, device availability, or runtime calibration data, in real-time systems.
For the intermediate representation, several MLIR~\cite{lattner2020mlir} dialects exist.
QSSA~\cite{qssa} and QIRO~\cite{qiro} expose the quantum data flow explicitly,
while further dialects are provided by the Catalyst compiler~\cite{catalyst}
and by CUDA-Q in the form of the Quake dialect~\cite{nvidia2024quake}.
In our work, we build on the Quake dialect, rather than using existing quantum pass pipelines.
The approaches discussed above motivate a unified selector, which jointly optimizes \textit{phase ordering} and \textit{device selection}.
This allows our scheduler to reason over compiler transformations and hardware constraints, while also adapting to live system status. 
Our paper motivates multi-objective optimization across fidelity, execution latency, and compilation effort, targeting more effective quantum execution strategies by addressing the tasks as a combined challenge.

\section{A Software Approach Toward Unified Scheduling Framework}\label{sec:unified_scheduler}
In this section, we present our approach towards a unified scheduling framework for optimizing the execution of quantum programs on \gls{nisq} devices.
While the long-term vision is to integrate \textit{Device selection} and \textit{Phase ordering} into a single intelligent system, currently we have two separate components that are planed to be integrated.
This module uses supervised learning techniques to make informed decisions about which quantum device to target, taking into account both program features and real-time device conditions. 
This is our foundational building block in the broader goal of constructing a unified scheduler, which is capable of making decisions across the entire quantum execution pipeline. 
In the following subsections, we first describe the supervised learning-based \textit{device selection} module and then outline our conceptual design and strategy for extending this work into a unified scheduling framework.

\subsection{Device-selection Component}\label{subsec:dev_selector}

We collected a training dataset for our \gls{ml}-based \textit{device selection} module.

Each of the data set's entries consists of the execution data of each input program on available devices, shot distribution results, and execution time on each device.
Additionally, the program is simulated extensively on a simulator to calculate its noiseless reference distribution.
To compute fidelity for labeling, we compare noisy execution results against the noiseless simulation using Hellinger fidelity~\cite{jarzyna2020geometric}, shown in Eq.~\ref{eq:hellinger-fidelity}, which is divided by the execution time duration.

In detail, we assume $\mathcal{K}$ to be the set of all discrete outcome indices (e.g., measurement bitstrings). The following equations show how the outputs are calculated and labeled for training the \textit{target device selection} component.

\begin{align}
    F_H(P,Q) &= 1 - H(P,Q)^2 = \sum_{k\in\mathcal{K}}\sqrt{\,p_k\,q_k\,}.\label{eq:hellinger-fidelity}%
\end{align}

where:
\begin{itemize}
    \item The Hellinger distance between $P$ and $Q$ is defined as
  \begin{equation}
    H(P,Q) = \frac{1}{\sqrt{2}}\sqrt{\sum_{k\in\mathcal{K}}\left(\sqrt{p_k}-\sqrt{q_k}\right)^2},
    \label{eq:hellinger-distance}
  \end{equation}
  
  which quantifies the dissimilarity between two discrete probability distributions~\cite{jarzyna2020geometric}.
    \item $p_k = \frac{SD_P[k]}{N_P}$ and $q_k = \frac{SD_Q[k]}{N_Q}$ are the normalized empirical probabilities for each outcome $k$.
  \item $SD_P[k]$ and $SD_Q[k]$ are the raw shot distributions of outcome $k \in \mathcal{K}$ from two probability distributions $P$ and $Q$  respectively (in this case $P$ is a device and $Q$ is the noiseless simulation).
  \item $N_P = \sum_{k \in \mathcal{K}} SD_P[k]$ and $N_Q = \sum_{k \in \mathcal{K}} SD_Q[k]$ are the total number of samples (shots) collected for $P$ and $Q$.
  \item $t_{\mathrm{exec}}$ is the total execution time (e.g., runtime of the quantum program on hardware).
  \item The final \textbf{score} is defined as the tuple fidelity and execution time:
  \[
  \text{\textbf{score}} = ({F_H(P, Q)},t_{\mathrm{exec}}).
  \]
\end{itemize}
The component's goal is to predict the Hellinger fidelity for each device and rank them by suitability for executing the program.
Our \textit{device selection} component applies various \gls{ml} methods, including \gls{mlp}, \gls{rf}, \gls{knn}, and \gls{svm}.

The \textit{device selection} component is structured into subroutines for dynamic labeling, training, and model evaluation.
This design enables flexible, efficient adaptation of models and labeling strategies with minimal overhead.
Importantly, the exact value of the Hellinger fidelity is of secondary concern. 
The key objective is to preserve the \textit{relative structure} across devices. 
Systematic prediction errors that scale all outputs by a common factor \(\eta\) are acceptable, as long as the \textit{ranking of devices} and \textit{relative score ratios} (e.g., the Hellinger fidelity for a given program and device, for a set of devices $\{d_i\}$: \(F_{Hd_1} > F_{Hd_2}\) and \(F_{Hd_1}/F_{Hd_2}\) ) remain consistent. We discuss the experimental results of our \textit{device selection} in Section~\ref{sec:experiments}.
\subsection{MQSS Pass Suite}\label{subsec:pass_suit}
We extend the  \gls{mqss} compilation pass suite with a significantly expanded set of passes
and make our fork openly available under \footnote{\url{https://anonymous.4open.science/r/XXXX-Passes-Suite-614D/README.md}}.
Together with the pass selector, this yields a clean separation between action selection
and execution across two languages: the passes are implemented in C++, while the
\texttt{pybind11} interface exposed by the selector allows selection policies to be
developed in Python, independently of the underlying pass implementations.
This design enables reproducible phase-ordering experiments.
Passes are categorized into three types:
\begin{itemize}
    \item \textbf{Cancellations:} Passes that simplify operation sequences by reducing
          them to the identity operation.
    \item \textbf{Transformations:} Passes that structurally rewrite the program while
          preserving its computational semantics (e.g., qubit mapping and control
          flow mutations).
    \item \textbf{Decompositions:} Passes that decompose operations into equivalent
          sequences of more primitive operations.
\end{itemize}
At the data generation time the suite had in total 92 passes, for supporting all necessary operations in the Quake dialect of CUDA-Q~\cite{nvidia2024quake} and a mapping pass. As the git project is under development the total number of passes may change.

\subsection{Phase-ordering Component}\label{subsec:pass_combination}
Manual optimization of compilation passes across many programs and devices is infeasible, especially in dynamic environments where devices change or new compilation passes are introduced. Each such change would require re-running an exhaustive \gls{dse}, making the approach impractical at scale.
\gls{rl} offers a possible solution by learning generalizable strategies that adapt to both program characteristics and backend configurations. Our \gls{rl} implementation focuses on the Actor-Critic algorithm~\cite{murphy2023probabilistic}. The model's architecture comprises  two neural networks:
\begin{itemize}
  \item \textbf{Actor network:} predicts the next compilation pass by applying a policy 
    \[ \pi(s_t, a_t)\,, \]
    where \(s\) denotes the current state and \(a\) is the chosen pass action. This learns the policy $\pi(s, a)$ that maps a given state $s$ (e.g., program representation and selected device context) to the action $a$ corresponding to the next compilation pass to apply.
    Additionally, the actor has the action \textit{"Finish"} in which the actor stops the compilation by choice and claims no further optimization is necessary.
  \item \textbf{Critic network:} estimates the state‐value function 
    \[ V(s_t)\,, \]
    which acts as feedback, by evaluating how advantageous it is to transition from state \(s\) to the next state \(s'\).
    Specifically, this network estimates the value function $V(s)$ to predict how promising the current state is in terms of achieving long-term rewards, given the current policy (e.g., program fidelity, compilation time, or depth). 
    Since the critic's target values depend on the current policy, which is itself being updated by the actor, both networks are effectively chasing a moving target --- a key source of training instability in Actor-Critic methods.
\end{itemize}

\begin{figure*}[t!]
    \centering
    \begin{subfigure}[t]{\columnwidth}
        \includegraphics[width=\textwidth]{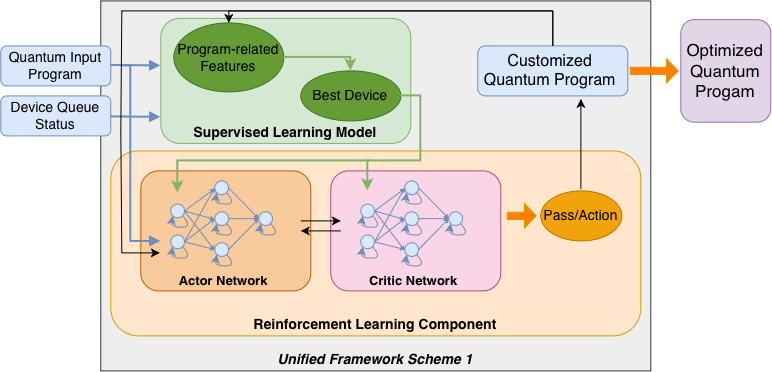}
        \caption{} 
        \label{fig:Unified_scheduler_scheme_V1}
    \end{subfigure}%
    \hspace*{10mm}%
    \begin{subfigure}[t]{\columnwidth}
        \includegraphics[width=\textwidth]{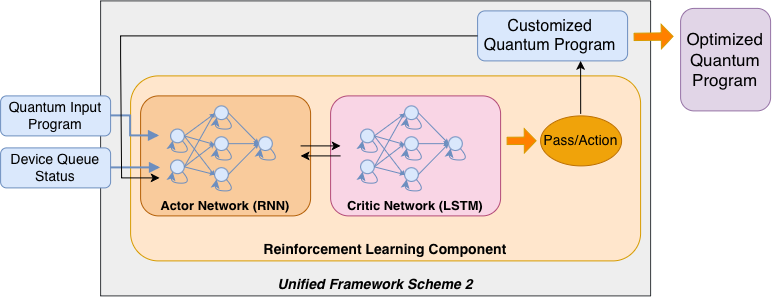}
        \caption{}
        \label{fig:Unified_scheduler_scheme_V2}
    \end{subfigure}%
    \caption{a) Conceptual design of the unified scheduler (V1) - Hybrid Supervised + \gls{rl}, where \textit{device selection} and \textit{phase ordering} are handled by interlinked modules. 
    b) Conceptual design of the unified scheduler (V2) - Fully \gls{rl}-based, combining \textit{device selection} and \textit{phase ordering} into a single action space for \gls{rl} training.}
    \label{fig:conceptual_designs}
\end{figure*}

Both models' input state includes a representation of the program (e.g., in an \gls{ir} or feature-based format) and a list of candidate devices selected by a prior \textit{device selection} module.

The output is the next compilation pass best suited to the given device and program characteristics, according to the actor's learned policy.

The \gls{rl} environment is defined by the set of currently available quantum devices and the collection of implemented compiler passes. This environment provides feedback (reward signals) based on the performance of the compiled program under different pass sequences, targeting objectives like fidelity, program depth, operation count, compilation time, or execution time. 
The state space defined by all devices,  arbitrary quantum programs, and different stages in the pass sequence is, in practice, infinitely large. To obtain stable and sample-efficient advantage estimates during critic training, we use \gls{gae}~\cite{murphy2023probabilistic}.
 In \textit{actor-critic} and other policy-gradient methods, a reliable estimate of the advantage
    function improves sample efficiency and stabilizes learning. The one-step \gls{td} -residual is defined as:
    \[
      \delta_t \;=\; r_t + \gamma V(s_{t+1}) - V(s_t).
    \]
    The \gls{gae}\((\gamma,\lambda)\) estimator mixes multi-step \gls{td} errors as
    \[
  \hat{A}_t^{\mathrm{GAE}(\gamma,\lambda)}
  \;=\; \sum_{l=0}^{\infty} (\gamma\lambda)^l \,\delta_{t+l}.
\]
\label{eq:gae_function}
    and provides a tunable bias–variance trade-off via \(\lambda\in[0,1]\). In practice, \gls{gae} is commonly used with \textit{actor-critic} algorithms and combined with
    advantage normalization and value-function regularization.
    
\subsubsection{Reward Functions}
The reward function design is grounded in three desirable 
properties.
These properties formalize our expectations of 
a compilation policy with respect to executability, 
structural reduction, and compilation efficiency.
\paragraph{Executability.}
An environment state $s$ is executable if all constraints imposed by the 
target hardware are satisfied, formally:
\[
e(s) = 1 \iff u(s) = 0 \land m(s) = 1,
\]
where $u(s)$ denotes the number of non-native operations 
and $m(s)$ indicates whether the program has a valid physical mapping. 
A reward function satisfies this property if executable states have a higher reward value than non-executable states.
\paragraph{Early stopping}
A reward function satisfies the early stopping property if, for any two rollout runs, reaching a terminal state at 
steps $t < t'$ respectively, with  equivalent executable final states $s_t = s_{t'}$, the terminal rewards satisfy:
\[
r_{\text{done},t} > r_{\text{done},t'}.
\]
This encodes the preference for shorter compilation paths when the final compiled program is equivalent.
\paragraph{Structural optimization}
A reward function satisfies the structural optimization property if it induces monotonic improvement with respect to a set of metrics. Namely, operation count and depth.
Let $d(s)$ denote program depth and $g(s)$ the total operation count. Then improvements are defined as:
\begin{equation}
\begin{split}
&\bigl(d(s_{t+1}) < d(s_t) \lor g(s_{t+1}) < g(s_t)\bigr) \\
&\quad\land\; \bigl(u(s_{t+1}) = u(s_t)\bigr)
\;\Rightarrow\; r_{t+1} > r_t
\end{split}
\end{equation}
where $u(s)$ denotes the number of unsupported operations.
Conversely, degradations in either metric do not yield positive reward contributions.
\begin{equation}
\begin{split}
&\bigl(d(s_{t+1}) \geq d(s_t) \land g(s_{t+1}) \geq g(s_t)\bigr) \\
&\quad\land\; \bigl(u(s_{t+1}) = u(s_t)\bigr)
\;\Rightarrow\; r_{t+1} \leq r_t
\end{split}
\end{equation}
\newline
\begin{equation}
    R_t=\phi_t -\phi_{t-1} + r_{done}
   \label{eq:potential} 
\end{equation}
Based on these properties, several reward functions were tested. The group of functions with the best result was the potential functions.
This principle is that these functions assign at each step a score to the state.
The reward value is defined by Equation \ref{eq:potential}.
In this group of potential reward functions, the best-performing candidate tested is given by 
\begin{align}
\phi_t &= c_0 \!\left(1 - \tfrac{d_t}{c_d d_{\text{init}}}\right)
        + c_1 \!\left(1 - \tfrac{g_t}{c_g g_{\text{init}}}\right)
        + c_2 \!\left(1 - \tfrac{u_t}{u_{\text{init}}}\right) \label{eq:phi}\\[4pt]
\phi_{\max} &= c_0 \!\left(1 - \tfrac{1}{c_d}\right)
             + c_1 \!\left(1 - \tfrac{1}{c_g}\right)
             + c_2 \label{eq:phi_max}\\[4pt]
\text{effort}_t &= 1 - \tfrac{\text{step}_t}{\text{step}_{\max}} \label{eq:effort}
\end{align}

\begin{equation}
r_{\text{done}} =
\begin{dcases}
  \dfrac{\phi_t + \text{effort}_t}{\phi_{\max}} \times 5
    & f{=}1,\; e^{(t)}{=}1 \\[6pt]
  -3
    & f{=}1,\; e^{(t)}{=}0 \\[2pt]
  0 & \text{otherwise}
\end{dcases}
\label{eq:r_done}
\end{equation}
where $d_t$, $g_t$, and $u_t$ denote the program depth, total operation count,
and unsupported operation count at step $t$, respectively.
$f \in \{0,1\}$ indicates whether the actor selected the finish action,
and $e \in \{0,1\}$ indicates whether the resulting program is executable.
The coefficients $c_0, c_1, c_2, c_d, c_g$ are heuristically derived scaling 
factors used to balance the contribution of each term.
The condition $u(s_{t+1}) = u(s_t)$ ensures that structural 
improvements are attributed independently of changes in 
executability, preventing conflation of the two reward contributions.
\newline
Since \textit{actor-critic} networks are separated, different hyperparameters and architectures offer room for further optimization.
RNN and LSTM models may improve predictions compared to standard \glspl{mlp} due to their ability to retain sequential state information. 
However, a baseline implementation using a standard \gls{mlp} will serve as our starting point for the evaluation.
It is now possible to estimate the lower bound of the probability for a successful compilation run, given a quantum program $n$ with unsupported operation types.
The actor can decide from a set of 94 actions: the 92 compilation passes, a
mapping pass that fits the program to the physical constraints of the target
device, and an action by which the actor signals it has finished its
compilation. The mapping pass is listed separately because it is the only pass that takes
input parameters, namely the connectivity constraints of the target device. With this set, the actor can compile for an arbitrary device and every native gate set that is supported by the Quake dialect.
For tractability of the estimation, the pass count is rounded up to 100.
A successful compilation is defined by being executable, as only if the program is executable, then the Finish action is positively rewarded.
It is also assumed that the $n+1$ actions are order-independent, which holds except for a small number of passes with explicit ordering dependencies, such as the $\texttt{Map} \to \texttt{SwapToCxCxCx}$ sequence.
For a rollout length of $k$ with $k\geq n+2$, the probability of a successful rollout is given by Eq.~\ref{eq:Finish_estimation}.
\begin{equation}
    P_{\text{success}} \approx
    \sum_{x=0}^{k-n-2}
    \underbrace{\prod_{i=0}^{n} \frac{n+1-i}{100}}_{P_{\text{decomposition+map}}} 
    \cdot 
    \underbrace{\left(\frac{100-(n+2)}{100}\right)^{x}}_{P_{\text{other Passes}}}
    \cdot 
    \underbrace{\frac{1}{100}}_{P_{\text{finish}}} 
    \label{eq:Finish_estimation}
\end{equation}
This equation highlights how unlikely a successful compilation rollout is under random action selection, which is a reasonable assumption for an untrained actor, thereby making the terminal reward signal difficult to discover through exploration alone.
To address this issue, we propose the DAEP approach.
In general, during a \gls{rl} approach, exploration and exploitation compete in order to find a near-optimal policy.
We add a third competitor, a guidance term in the loss function, which decays faster than the two other components and induces a constraint in the search space for policies.
\begin{figure}[h]
\centering
\begin{tikzpicture}[
    force/.style={-Latex, thick, line width=1.5pt},
    label/.style={font=\small\bfseries, align=center}
]

\node[draw, circle, minimum size=1.2cm, fill=gray!15] (agent) at (0,0) {Agent};

\draw[force, blue!70!black] (0, 2.8) -- (agent);
\node[label, blue!70!black] at (0, 3.3) {DAEP \\ \small(Guidance)};
\node[font=\footnotesize, blue!50!black] at (1, 2.6) {$p_{DAEP} \to 0$};
\node[font=\footnotesize, blue!50!black] at (1, 2) {$\alpha_{DAEP} \to 0$};

\draw[force, green!60!black] (-2.8, -1.8) -- (agent);
\node[label, green!60!black] at (-3.2, -2.3) {Entropy \\ \small(Exploration)};
\node[font=\footnotesize, green!50!black] at (-2.2, -0.8) {$\alpha_{H} \to 0$};

\draw[force, red!70!black] (2.8, -1.8) -- (agent);
\node[label, red!70!black] at (3.2, -2.3) {Advantages \\ \small(Exploitation)};
\node[font=\footnotesize, red!50!black] at (2.2, -0.8) {$\eta \to \delta$};

\end{tikzpicture}
\caption{Decay schedules over training epochs. The DAEP trigger probability
$p_{DAEP}$, the DAEP magnitude coefficient $\alpha_{DAEP}$ and the entropy
coefficient $\alpha_{H}$ decay to zero, so that guidance and exploration
vanish and the actor converges on exploiting the critic's estimates. The
learning rate $\eta$ decays to a $\delta > 0$ and thus remains active
throughout training.}
\label{fig:warmstart_triangle}
\end{figure}
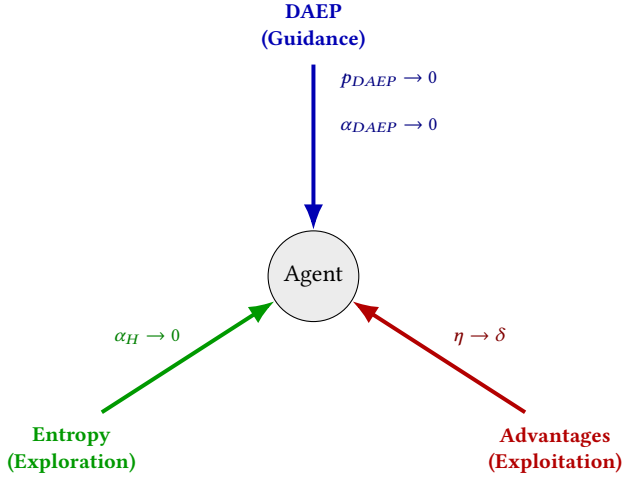

In regular Actor-Critic schemes, the total actor loss is calculated from the
policy-gradient term, which is scaled by the critic-based advantage estimate,
and an entropy term.
The entropy term is calculated from the action distribution's entropy weighted
by a coefficient $\alpha_{H}$.
Both the entropy coefficient and the learning rate $\eta$ decay during training
so that the actor can converge on a policy.
The new guidance component is integrated via two different coefficients: one
determines the trigger probability during a rollout, $p_{DAEP}$, the other,
$\alpha_{DAEP}$, is the magnitude coefficient. Both decay on a faster schedule
than the learning rate.
During a rollout, the actor's selection gets overwritten according to a
predefined guidance policy with probability $p_{DAEP}$.
This introduces a new loss component, namely a negative log-likelihood term,
which is evaluated only on those steps at which the guidance policy actually
replaced the actor's selection, and which is added to the total loss with the
coefficient $\alpha_{DAEP}$, analogous to standard behavior cloning.
The two coefficients separate \emph{how often} guidance is applied from
\emph{how strongly} the policy is pulled towards it.
A guided action is therefore always executed, but not necessarily reinforced:
as $\alpha_{DAEP}$ decays, the imitation signal weakens relative to the
advantage term, and the actor may assign a higher probability to a different
action than the one prescribed by the guidance policy.
This yields a gradual transition away from the guidance policy rather than the
abrupt one a single coefficient would produce.

\subsection{Unified Scheduling Framework}\label{subsec:unified_framework_design}

To address the bootstrapping dilemma, described in Section~\ref{sec:prob_statement}, we propose a unified scheduling framework with two trial-error bootstrapping schemes in which both device selection and pass selection are optimized jointly.
The overall idea is to frame the problem as a \gls{rl} task, where the agent interacts with an environment defined by:
\begin{itemize}
    \item \textit{State}: quantum program-related features (e.g., qubit count, depth), current backend status (e.g., calibration data, queue length), and any intermediate compilation results.
    \item \textit{Action}: \textit{device selection}, \textit{pass selection}, or both simultaneously.
    \item \textit{Reward}: performance metrics such as execution fidelity, compilation overhead, and latency.
\end{itemize}
In this context, we can let target devices be selected randomly at the starting stage, then get feedback from the \gls{rl} environment. In later stages, we can reuse input data and update the estimated label.
The \gls{rl} agent is implemented as an \textit{actor–critic} model, with the \textit{actor} network learning a policy to select the next action (device and/or pass) and the \textit{critic} network estimates the value of the current state. 

We currently consider two conceptual schemes for developing such a unified scheduler:
\begin{itemize}
    \item \textit{Approach 1}  – Hybrid Supervised + \gls{rl}: shown in Fig.~\ref{fig:Unified_scheduler_scheme_V1}, both \textit{device selection} and \textit{pass selection} are delegated to different sub-architectures.
    The intermediate results from one component (\textit{device selection} or \textit{pass selection}) are used to refine the training data at each training epoch.
    The output of the trained \textit{device selection} module is provided as part of the state input to another \gls{rl}-based module tailored for \textit{phase ordering}.
    \item \textit{Approach 2}  – Fully \gls{rl}-based: shown in Fig.~\ref{fig:Unified_scheduler_scheme_V2}, \textit{device selection} and \textit{phase ordering} are combined directly to form a new, larger action set. 
    The \gls{rl} agent then learns to choose from this combined action set.
\end{itemize}
Overall, they illustrate two possible paths toward a unified scheduler: 
one retaining modular specialization between \textit{device selection} and \textit{phase ordering} (Approach 1) and one performing joint optimization in a single policy (Approach 2).

\section{Experiments}\label{sec:experiments}
\subsection{Device selection}
\begin{figure*}[t]
    \centering
    \includegraphics[width=\textwidth]{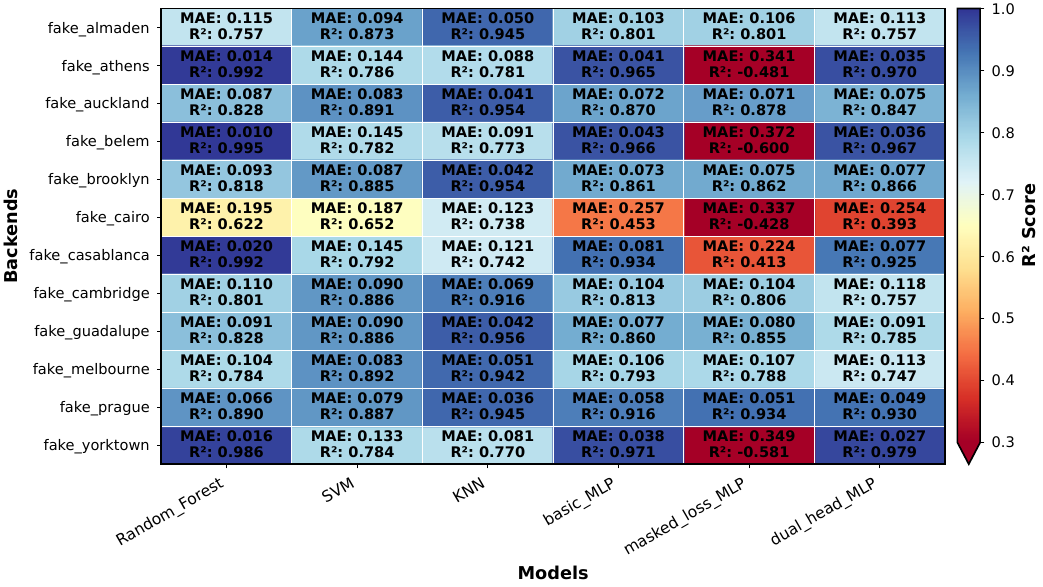}
    \caption{Prediction analysis on initial models for the 128 MQT-Bench program set}
    \label{fig:heatmap}
\end{figure*}
For the initial experiments on \textit{device selection}, we use two different data sets.
The first set consists of 128 programs generated by \textit{MQTBench}~\cite{quetschlich2023mqtbench}, a component of \gls{mqt}. These programs were executed on twelve mock devices from \texttt{ibm-fake-backend} provider, with different connectivity characteristics and native gate sets.
The second set contains 2,880 random benchmarking programs for the qubits 1-12, 6 folds and various samples on two real devices.

For both sets, we use an already optimized compiler 
for program compilation, and we use a 
 set of models 
 to determine good candidates for the model. 
 In particular, we chose an \gls{mlp}, an \gls{rf}, a 
 \gls{svm}, and a \gls{knn}.
 For each \gls{ml} approach, we use a short grid search for optimization.
By analyzing the output of the \texttt{basic\_MLP} it could be seen that one challenge was to predict programs that are infeasible to execute on the respective mock device (e.g., when the job needs more qubits than available on the device).
One of two approaches for improvement consists of modifying the loss function of the \texttt{masked\_loss\_MLP}, by adding weights to the loss with respect to the ground truth labels.
The second is a \texttt{dual\_headed} approach, here one head predicts feasibility and the other head predicts the expected Hellinger fidelity shown in Eq.~\ref{eq:hellinger-fidelity}.
Eq.~\ref{eq:MAE} displays the \gls{mae} for each model's prediction. 
Note that \gls{mae} alone is not expressive enough, as it averages errors and does not distinguish between many small errors and few large deviations. 
\begin{equation}
    \text{MAE} = \frac{1}{n} \sum_{i=1}^{n} \left| y_i - \hat{y}_i \right|
    \label{eq:MAE}
\end{equation}
Therefore, we additionally use the $R^2$ Eq.~\ref{eq:R2} metric. It takes larger deviations on single samples into account by calculating the ratio of \gls{mse} to variance.
\begin{equation}
    R^2 = 1 - \frac{\sum_{i=1}^{n} (y_i - \hat{y}_i)^2}{\sum_{i=1}^{n} (y_i - \bar{y})^2},
    \label{eq:R2}
\end{equation}

where:
\begin{itemize}
    \item \( y_i \): The true value for the $i$-th sample.
    \item \( \hat{y}_i \): The predicted value for the $i$-th sample.
    \item \( \bar{y} = \frac{1}{n} \sum_{i=1}^{n} y_i \): The mean of the true values.
\end{itemize}
Fig.~\ref{fig:heatmap} shows the results on the \gls{mqt} dataset.
It displays the model's fidelity predictions for the 128 \gls{mqt}-bench programs on the chosen twelve backends.

The results show a worse performance by
\gls{mlp}, %
compared to the other models. %
We expect fine tuning to improve the predictions. %
While \gls{svm} and \gls{knn} show good results for fidelity predictions, they may not be a fitting choice for algorithms in the unified scheduler. Their good performance is based on the limited data set. In a full deployment, we expect an even more diverse program set and a huge variance in program depths and, hence, we anticipate that these algorithms will perform less well in the scheduler in practice than they did in these experiments.
That makes \gls{rf} the currently most suitable candidate on this program set.
\begin{figure}[t]
    \centering
    \includegraphics[width=0.8\columnwidth]{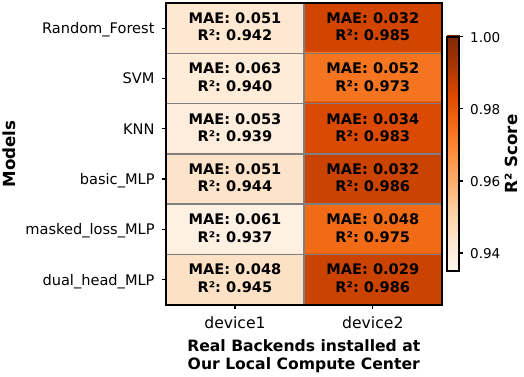}
    \caption{Prediction analysis on initial models for the 2,880 RB program set}
    \label{fig:heatmap_real}
\end{figure}
Fig.~\ref{fig:heatmap_real}  shows the result gained from executing the second data set on two real devices installed at our local compute center.
We can see that here all models perform nearly equally in \gls{mae} and $R^2$.
The reason for this is that all programs of the set could be executed on all devices making the training better suited.
This can be seen in the fact that there is nearly no change between basic, masked\_loss and dual\_head \gls{mlp}.
The difference across the devices can be related to device properties and environmental influences of the real devices.
\begin{table*}[!t]
\centering
\renewcommand{\arraystretch}{1.2}
\begin{tabular}{l l r r}
\toprule
\textbf{Category} & \textbf{Metric} & \textbf{Guidance (DAEP)} & \textbf{No Guidance} \\
\midrule
\multirow{6}{*}{\textbf{Finish Action Analysis}}
 & Samples             & 154       & 154      \\
 & Avg. total rewards  & 4.32      & 0.44     \\
 & Finish called       & 153/154   & 0/154    \\
 & Finish $+$            & 146       & 0        \\
 & Finish $-$            & 7         & 0        \\
 & Executable programs & 95.5\%    & 42.1\%   \\
\midrule
\multirow{4}{*}{\textbf{Depth and Operations}}
 & Avg. Depth change   & +115.28   & +134.52  \\
 & Avg. Operation change    & +324.62   & +375.80  \\
 & Depth reduced       & 87/154    & 8/154    \\
 & Ops. reduced        & 93/154    & 15/154   \\
\midrule
\multirow{7}{*}{\textbf{Action Classification}}
 & TP                  & 253       & 870      \\
 & FP                  & 9         & 1044     \\
 & TN                  & 38        & 1688     \\
 & FN                  & 881       & 765      \\
 & Precision           & 96.56\%   & 45\%     \\
 & Recall              & 22.31\%   & 53\%     \\
 & F1                  & 36.25\%   & 49\%     \\
\bottomrule
\end{tabular}
\caption{Chemistry Dataset – Validation Summary comparing DAEP and No Guidance models}
\label{tab:chem_combined}
\end{table*}
\subsection{Pass selection}
In order to train and evaluate the Pass selection component we used a data set containing the Hamiltonians of chemical molecules encoded in quantum programs.
This data set had 1548 samples and 10\% are used as validation sample. Furthermore to ascertain that this is no special case the models were trained for two different Devices with different mappings and native operation sets.
A comparison with other compilers such as Qiskit would only be meaningful at this stage if the action set mirrors the exact pass set available for such a compiler.
Therefore we introduce a tailored set of metrics to evaluate the actor and critic networks.
The actor can decide when it is finished with the compilation on a quantum program by calling the finish action, based on the executability either a positive reward is awarded or a negative reward if the program is not executable when the actor finishes.
Based on this we can measure how efficient the actor is, as it illustrates the
actor's ability to recognize whether this core requirement is fulfilled.
The next part of the metric set captures information on how much the programs increase on average in depth and operation count. 
Additionally it shows in how many samples action were taken that reduced the operation count or the depth.
Thereby revealing how many actions the actor took that optimized the program.
A final metric is a pseudo classification of the critic's output.
By comparing the critic's output with the analytical value calculated via GAE algorithm and the reward values, it is possible to define the class True positive (TP).
Every prediction that only deviates by a small constant $\delta d$ from the analytical value is in this class, as long as the signs match.
If the deviation  is larger than the constant the prediction is categorized as True negative (TN).
For mismatching signs accordingly False negative (FN) is a negative prediction for a positive value and vice versa False positive (FP).
The previously mentioned constant $\delta d$ is not defined over a fraction of the ground truth value as these values can fluctuate depending on the position of the sequence.
Therefore we chose a constant value, the disadvantage of this approach lies in the sensitivity for the classes not only in cases where the sign mismatches but in general.
It directly determines the boundaries for each class.
Together this set of metrics informs about actor's and critic's performance over all verification samples.
\newline
Table \ref{tab:chem_combined} shows the results of our Pass selection component achieved according to the metric set on data set of Hamiltonians of molecules encoded in quantum programs.
The actor learned to terminate when further pass application was not expected to improve the program.
In the one remaining sample, the finish action was not called, which can be attributed to the action horizon limit --- rollouts exceeding 20 steps are truncated.

Similarly in 7 cases, samples used gates which the actor was not familiar. This can be explained by the limited number of samples in data set (1548). 
Furthermore the actor with the DAEP approach yields better compilation performance, as reflected in the smaller resulting depth and operation count on average.
This is shown by comparing the values in the table's columns. Another fact illustrated is, that the critic needs more training.
In the guided model the action classification is very conservative, but compared to the model without guidance the critic performs better.
A reason for this can be found in the number of actions taken, this leads to more training for the critic.
The fundamental reason is that with guidance the actor learns to finish
efficiently with the least amount of actions, whereas in the other case the
rollouts nearly always run until the end of the action horizon.

\section{Conclusions} \label{sec:conclusion}
In this work, we introduced the concept of a \textit{unified scheduler} and presented our approach for its implementation.
Our theoretical analysis demonstrated the feasibility of \textit{jointly selecting a Pareto-optimal combination of a quantum device together with a sequence of compilation and optimization passes} for heterogeneous \gls{hpcqc} software stacks.
To achieve an efficient and practical solution, our approach leverages \gls{ml} algorithms.
We presented experimental results that demonstrate the feasibility of both the \gls{ml} component on \textit{device selection} and the component on \textit{pass selection}.

The data provided empirical evidence that \textit{device selection} can be effectively addressed using \gls{ml} techniques.
In contrast, the experimental data on the pass selection component suggests
that \gls{rl} without additional methods is not a viable approach for training
a network for compilation, at least not under the resource constraints
considered here.
Training the pass selection networks for different devices, suggest that this approach does not only work for specific native operation sets or physical constraints.
By introducing a decaying guidance component during training --- overwriting the actor's actions with decreasing probability and magnitude --- DAEP accelerates convergence and improves reward exploitation.
Through the properties we defined, we can be certain that the actor learns a policy that leads to our intended objective.
The data also reveals the need for more extensive training for the pass selection component.
The combination into a unified scheduler introduces additional complexity, 
but also a key opportunity: by exploiting the interdependency between device 
selection and compilation, the unified approach has the potential to outperform 
pipelines that optimize each stage independently.
Furthermore, the \gls{ml}-based unified scheduler supports continuous improvement 
of predictions through online training on new data acquired during execution.

\begin{acks}

 This work is supported by the German Federal Ministry of Research, Technology and Space (BMFTR) with the grants 13N15689 (DAQC), 13N16063 (Q-Exa), 13N16078 (MUNIQC-Atoms), 13N16187 (MUNIQC-SC), 13N16690 (Euro-Q-Exa), 13N16894 (MAQCS), European fundings 101136607 (CLARA), 101114305 (Millenion), 101113946 (OpenSuperQPlus), 101194491 (QEX), and the Bavarian State Ministry of Science and the Arts through funding, as part of \gls{mqv}, Q-DESSI.
\newline
 We would also like to express our gratitude to Aleksandra Swierkowska for her support.
\end{acks}

\clearpage
\bibliographystyle{ACM-Reference-Format}
\setcitestyle{sort}
\bibliography{references}

\end{document}